\documentclass[]{fairmeta}

\usepackage{MyPreamble}

\newcommand\OursFull{Sparse Co-Engagement \underline Graph Schema for \underline{G}enerative \underline{Rec}ommendation}
\newcommand\Ours{G2Rec}
\newcommand\Metas{Meta's}
\newcommand\AtMeta{ at Meta}

\title{Structuring and Tokenizing Distributed User Interest Context for Generative Recommendation}

\author[1,*]{Ruizhong Qiu}
\author[2]{Yinglong Xia}
\author[2]{Dongqi Fu}
\author[2]{Hanqing Zeng}
\author[2]{Ren Chen}
\author[2]{Xiangjun Fan}
\author[2]{Hong Li}
\author[2]{Hong Yan}
\author[1]{Hanghang Tong}

\affiliation[1]{University of Illinois Urbana--Champaign}
\affiliation[2]{Meta MRS}

\contribution[*]{Work done during an internship at Meta}

\abstract{
Generative recommendation is an emerging paradigm that has shown promise in industrial recommendation systems, aiming to predict users' next interactions from their historical behaviors.
At the core of generative recommendation lies item tokenization, which bridges item semantics and recommendation models.
However, existing methods often struggle to effectively \textit{organize} and \textit{inject} complex user-behavioral and item-semantic contexts into recommendation models simultaneously.
On the one hand, existing graph-based integration methods, such as graph serialization and graph neural networks, either suffer from scalability issues or exploit only local graph information.
On the other hand, existing semantic tokenization methods typically rely on heuristics and lack explicit supervision signals, which may lead to inaccurate or suboptimal semantic representations.
To address these limitations in user interest context modeling, we propose \textbf{\Ours{}}, a scalable framework that unifies holistic graph-based user co-engagement modeling with semantic tokenization for industrial-scale generative recommendation.
\textbf{First}, we construct a sparsified \emph{item-item co-engagement graph} with $\OP O(M\log M)$ edges as the item schema, where $M$ denotes the total number of interactions.
\textbf{Second}, we design a scalable ``soft'' clustering algorithm with time complexity $\OP O(\rho M\log M)$ per iteration to extract distributed \emph{interest prototypes} from the constructed graph, where $\rho$ is a small constant representing the sparsity of the soft cluster membership distribution rather than a hard one-to-one assignment.
\textbf{Third}, based on the \emph{interest prototypes} and \emph{item interest profiles} extracted by soft clustering, we tokenize these profiles together with users' interested items to train a generative sequential recommendation model.
Overall, \Ours{} enables recommendation models to capture holistic and semantically grounded user interest prototypes without requiring ground-truth user interests, thereby providing more comprehensive and accurate modeling of user behavior contexts in industrial sequential recommendation.
Online deployment across product surfaces and extensive experiments on public datasets demonstrate the superiority of \Ours{} over existing methods.
}

\metadata[Keywords]{Generative Recommendation, Tokenization, Graph Clustering}

\begin{document}

\maketitle

\section{Introduction}

Generative recommendation (GR) \citep{geng2022recommendation,gao2023chat,zhu2023large,hou2024large,deng2025onerec} has emerged as a promising paradigm in the field of industrial recommender systems. By leveraging the autoregressive nature of user behavior, GR aims to predict the next interactions of users based on their historical contexts using large language model (LLM) architectures, thereby providing users with a more personalized and responsive experience. This approach has shown considerable potential in various scenarios, including e-commerce, online advertising, and content streaming, where understanding the sequential nature of user behavior is crucial for delivering relevant recommendations. By modeling the temporal dependencies between user interactions, GR seeks to capture the underlying preferences and interests that drive user behavior, ultimately leading to more accurate and effective recommendations.

At the core of GR is item tokenization to bridge item semantics and the GR model. However, existing GR methods often struggle to effectively incorporate complex user-behavioral and item-semantic contexts simultaneously into the recommendation model.
On the one hand, while there exist graph integration methods to leverage the relational information from user-item relational graphs, they suffer from scalability issues or only utilize local graph information. For instance, graph serialization methods transform the graph into a sequence, but the long serialization can incur expensive computation of LLMs. Graph neural networks (GNNs), on the other hand, typically perceive only a small subgraph for each user, but they cannot fully exploit the holistic graph information. These limitations hinder the effectiveness of graph-based methods in industry-scale generative sequential recommendation.
Another parallel line of work is \emph{semantic tokenization}, which aims to represent each item as a few semantic tokens. However, they typically rely on heuristic learning objectives and lack supervision signals for semantic tokens. Consequently, these methods may not be guaranteed to learn accurate semantic representations, leading to suboptimal performance in recommendation tasks.

To address these critical limitations in this complex context modeling, we propose \emph{\OursFull{}} (\textbf{\Ours{}}), a novel method that bridges holistic graph-based co-engagement modeling and semantic tokenization in industry-scale generative sequential recommendation.
Systematically, our method \Ours{} first constructs an item-item co-engagement graph, capturing the complex relationships between user behavior and items with empirical efficiency and theoretical accuracy.
We then design a ``soft'' clustering algorithm in \Ours{} to obtain item interest prototypes from the constructed graph, where each cluster represents an interest prototype, and each item serving as a node has the soft cluster membership distribution over clusters (but not a hard-assigned one-to-one exclusive membership).
Finally, based on the item interest profiles extracted from the soft clustering, \Ours{} tokenizes the item semantics along with the original user engagement sequence to train the generative sequential recommendation models.

Our proposed method \Ours{} offers several advantages over existing methods. Firstly, it provides a more comprehensive and accurate modeling of user behavior by incorporating both graph-based co-engagement contexts and semantic item representations. Secondly, our method is highly scalable, allowing for efficient processing of large-scale graphs and user sequences. Finally, the use of clustering algorithms ensures that the learned semantic tokens are accurate and meaningful, without relying on explicit labels or heuristics.
We conduct extensive experiments on public datasets and online deployment on product surfaces to evaluate the performance of G2Rec. The results demonstrate the superiority of our method over many existing generative recommendation methods, highlighting its potential for real-world applications in industry-scale generative recommendation.

\vspace{0.5em}\noindent\textbf{Main contributions.} We summarize the main contributions of this work as follows.
\begin{itemize}
\item\textbf{Sparse graph schema as co-engagement contexts.} We propose a sparsified item co-engagement graph as the item schema.  We show that sampling $\OP O(M\log M)$ co-engagements suffices to approximately preserve structural information (\THMref{preserv-lap}), where $M$ is the total number of interactions. This is substantially sparser than the original co-engagement graph of quadratic size $\OP O(M^2)$.
\item\textbf{Scalable soft graph clustering.} We propose a scalable soft graph clustering algorithm. It is based on our proposed differentiable modularity objective function, which can be computed in nearly linear time $\OP O(\rho M\log M)$ and can be accelerated by GPUs, where $\rho$ is a small sparsity constant.
\item\textbf{Clustering-based item interest profiling.} We propose (i) using the clusters of the item-item co-engagement graph as item interest prototypes  and (ii) leveraging the soft cluster membership of each item as the item interest profile for user behavioral information beyond item feature similarity.
\item\textbf{Interest profile tokenization for GR.} We introduce a novel representation of user interaction history sequences: alternating between items and item interest profiles, allowing the generative recommender to effectively learn user interest transition patterns.
\item\textbf{Strong online performance.} Large-scale online A/B tests on \Metas{} product surfaces show the superiority of \Ours{}. 
\end{itemize}

\section{Preliminaries}

\subsection{Notations}\label{ssec:prelim-notat}
Let $\mathcal{U}$ denote the complete set of users, where each user is represented as $u \in \mathcal{U}$. Similarly, let $\mathcal{I}$ denote the universe of items, with $i \in \mathcal{I}$ denoting a particular item. Each item $i\in\mathcal I$ has an embedding vector $x_i\in\BB R^d$, where $d$ denotes the embedding dimensionality. Let $X:=[x_1,\dots,x_{|\CAL I|}]^\top\in\BB R^{|\CAL I|\times d}$ denote the item embedding matrix. For each user $u\in\mathcal U$, we record the past interactions of the user $u$ as a sequence $\mathcal I_u=[i_1,\dots,i_N]$ of items, where $N$ denotes the number of items that user $u$ has interacted with in the past, and each $i_j\in\mathcal I$ $(j=1,\dots,N)$ denotes an item that user $u$ has interacted with. We call sequence $\CAL I_u$ the \emph{interaction history} of user $u$ and assume $\bigcup_{u\in\CAL U}\CAL I_u=\CAL I$ to avoid cold start. Let $M:=\sum_{u\in\CAL U}|\CAL I_u|$ denote the total number of interactions. These interactions can be of various forms (such as clicks, views, comments, and purchases), and they reflect the user behavioral patterns to be learned by the recommender system.

\subsection{Problem Definition}\label{ssec:prelim-def}
The objective of a generative recommender system is to predict an item that a user will probably interact with next, from a sequence of items that the user has interacted with in the past. The task is formally defined in Problem~\ref{PRB:rec}.

\begin{PRB}[generative sequential recommendation]\label{PRB:rec}
\textbf{Input:} (i) a user $u\in\mathcal U$; (ii) an integer $N$ specifying the length of the interaction history of user $u$; (iii) a sequence $\mathcal I_u=[i_1,\dots,i_N]$ of items that user $u$ has interacted with in the past; (iv) the embedding $ x_i$ of each item $i\in\mathcal I_u$. \textbf{Output:} rank the candidate items $\mathcal I$ by the likelihood of each item to be the next interaction of user $u$.
\end{PRB}

\section{Proposed Framework: \Ours{}}
\label{sec:method}
In this section, we introduce our proposed method \emph{\OursFull{}} (\Ours{}). We describe our graph construction procedure in Section~\ref{ssec:method-graph}, our new scalable soft graph clustering method in Section~\ref{ssec:method-clus}, and our graph-based co-engagement tokenization in Section~\ref{ssec:method-rec}. 
We will discuss industry-scale deployment in Section~\ref{ssec:method-deploy}.

\subsection{User Behavior Modeling via a Sparsified Item Co-Engagement Graph}\label{ssec:method-graph}
In this subsection, we describe how we construct a graph to capture user behavioral patterns. We will use the graph for item interest profiling in Section~\ref{ssec:method-clus}. 

In existing graph-based recommendation methods, user behavior is often modeled using the user-item bipartite graph. In the bipartite graph, users and items are represented as nodes, and edges connect users to the items they have interacted with.
However, for industry-scale applications with massive user bases, this can be extremely computationally expensive, especially when the user has a long interaction history. 

\vspace{0.5em}\noindent\textbf{Co-engagement graph.} To address this critical limitation, we propose eliminating user nodes from the constructed graph and model user behavior using an item-item graph instead. By eliminating users from the graph, we can focus on the more static item relationships, reducing the need for frequent updates. In this work, we propose the item-item \emph{co-engagement} graph as a natural definition of the item-item graph. The two items $i,j\in\mathcal I$ are said to be \emph{co-engaged} if there exists a user $u\in\CAL U$ such that $u$ has interacted with both items $i$ and $j$ (i.e., $i\in\CAL I_u$ and $j\in\CAL I_u$); the user $u$ here is not necessarily unique. The co-engagement graph captures the patterns of items being interacted with together, allowing us to model user behavior without explicitly representing users.

Formally, in the item-item co-engagement graph $\CAL G^*=(\CAL I,\CAL E^*)$, items $\CAL I$ serve as nodes of $\CAL G^*$, and co-engagements $\CAL E^*$ serve as edges of $\CAL G^*$. Let $\CAL E^*\subseteq\CAL I\times\CAL I$ denote the set of item-item co-engagements:
\AL{
\CAL E^*:={}&\bigcup_{u\in\CAL U}\CAL I_u\times\CAL I_u
={}\{(i,j):\exists u\in\CAL U\text{ s.t. }i\in\CAL I_u,\,j\in\CAL I_u\},
}
where $\times$ denotes the Cartesian product between sets. The item-item co-engagement graph $\CAL G^*$ provides a powerful tool for modeling user behavior patterns beyond item feature similarity. For example, consider two items: a smartphone and a phone case. While they may not be similar in terms of their features, they are often co-engaged by users who purchase a new phone and then look for accessories to go with it. The co-engagement graph can capture such patterns, providing insightful information that complements item feature information.

From a theoretical perspective, the interaction history sequence of each user can be represented as a path on the item-item co-engagement graph, as formally stated in \OBSref{express}. Hence, edges on the item-item co-engagement graph encode the user interest transition behaviors. 

\begin{OBS}[expressiveness]\label{OBS:express}
For any user $u\in\CAL U$, their interaction history sequence $\CAL I_u=[i_1,\dots,i_N]$ is a path on the co-engagement graph $\CAL G^*=(\CAL I,\CAL E^*)$. That is, $(i_t,i_{t+1})\in\CAL E^*$ for all $t=1,\dots,N-1$.\hfill$\square$
\end{OBS}

\noindent\textbf{Graph sparsification.} Nonetheless, since $|\CAL E^*|$ can be as large as $\OP O(M^2)$, using the original graph $\CAL G^*$ would be computationally prohibitive in industrial scenarios. To ensure scalability, we propose a theoretically grounded approach to sparsifying the graph. Our goal is to approximately preserve the graph Laplacian, which encodes essential structural information of the graph \citep{ng2001spectral}. Formally, given $m\ll N^2$, we define the sparsified co-engagement graph as $\CAL G:=(\CAL I,\CAL E)$ where we sample only $m$ co-engagements to include in $\CAL E$ for each user:
\AL{\CAL E:=\bigcup_{u\in\CAL U}\OP{sample}(\CAL I_u\times\CAL I_u,\,m),}
where $\OP{sample}(\,\cdot\,,\,m)$ denotes sampling to size $m$ with replacement to form an undirected graph. We will use the sparsified graph $\CAL G$ for item interest profiling in Section~\ref{ssec:method-clus}.

As implied by the following \THMref{preserv-lap}, we theoretically show that  $|\CAL E|=\OP O(M\log M)$ suffices to preserve the graph Laplacian, which is nearly linear w.r.t.\ the total number $M$ of interactions. Notably, this is substantially sparser than $|\CAL E^*|=\OP O(M^2)$.

\begin{THM}[nearly linear complexity]\label{THM:preserv-lap}
Suppose that every user has $N$ interactions. Let $L_{\CAL E^*}$ be the Laplacian of $\CAL E^*$, and let $L_{\CAL E}$ be the Laplacian of $\CAL E$ with edge weights $\frac{N(N-1)}{2m}$ to match the total edge weights of $L^*$. Given any $0<\delta<1$ and $\epsilon>0$, if we use
\AM{m:=\Big\lceil 2N\Big(\frac{1}{3\epsilon}+\frac1{\epsilon^2}\Big)\log\frac{2|\CAL I|}{\delta}\Big\rceil,}
then $|\CAL E|\le\OP O(M\log M)$, and with probability at least $1-\delta$,
\AM{(1-\epsilon)L_{\CAL E^*}\preceq L_{\CAL E}\preceq(1+\epsilon)L_{\CAL E^*},}
where $\preceq$ denotes the Loewner order.
\end{THM}



\subsection{Item Interest Profiling via Soft Graph Clustering}\label{ssec:method-clus}
In this subsection, we elaborate on how we leverage the item-item co-engagement graph to compute item interest profiles.

A key challenge in processing the co-engagement graph $\CAL G$ is that some co-engagement relationships can be random. To address this challenge, we propose to employ graph clustering to extract meaningful co-engagements. The intuition is that if a group of items is frequently co-engaged with each other (i.e., they belong to the same cluster in the graph), their co-engagements are more likely to be meaningful. Hence, we propose to use the clusters of the co-engagement graph $\CAL G$, which should provide reliable information for the recommendation model to learn. We call these clusters \textbf{item interest prototypes}, which can be interpreted as the underlying interest categories. Furthermore, as shown in Observation~\ref{OBS:express}, the co-engagement graph encodes user interest transition behaviors. Therefore, by treating each cluster as an item interest prototype, we can help the recommendation model to gain insights into the underlying interest transitions that drive user behavior.

However, since existing graph clustering algorithms (e.g., Louvain \citep{louvain}, Leiden \citep{leiden}) typically assume each node belongs to only one cluster, they are unsuitable for capturing complex interests where each item can correspond to multiple interest prototypes simultaneously. For instance, a short-form video of an avocado toast recipe may attract both users interested in ``healthy eating'' and those interested in ``morning routines''. To bridge this gap, we propose to employ \textbf{soft graph clustering} instead, which allows for continuous cluster memberships and differentiable optimization.

Nevertheless, existing soft graph clustering methods (e.g., \citet{yu2005soft}) are typically not sufficiently scalable for industrial use cases. To enable scalable soft graph clustering, we propose a differentiable objective for soft graph clustering that generalizes the classic notion of graph modularity \citep{newman2006modularity}, which allows end-to-end gradient-based optimization and can thus be solved efficiently on GPUs.

For each item \( i \in \mathcal{I} \), we aim to find a membership distribution $p_i\in\BB R^C$ , where $C$ denotes the number of interest prototypes, and $p_{i,a}$ represents the probability that item $i$ belongs to interest prototype $a\in\{1,\dots, C\}$. We call $p_i$ the \textbf{item interest profile} of item $i$. To ensure \( \sum_{a=1}^C p_{i,a} = 1 \), we use parameterization $p_i:=\OP{softmax}(z_i)$, where logits $z_i\in\BB R^C$ are the variables to be optimized. We initialize and sparsify the variables using the Leiden algorithm \citep{leiden}. Collectively, $p_1,\dots,p_{|\CAL I|}$ form a \emph{soft membership} matrix \( P:=[p_1,\dots,p_{|\CAL I|}]\Tp \in \mathbb{R}^{|\CAL I| \times C} \). 

Next, we introduce our differentiable objective for soft graph clustering. Let $A\in\{0,1\}^{|\CAL I|\times|\CAL I|}$ denote the adjacency matrix of the co-engagement graph $\CAL G$, and let $k\in\BB N^{|\CAL I|}$ denote the degree vector (i.e., $k_i := \sum_{j} A_{i,j}$ is the degree of each item $i\in\CAL I$). Recall that given a hard (cluster) membership $h\in\{1,\dots,C\}^{|\CAL I|}$ (i.e., $h_i$ is the interest prototype that item $i$ belongs to), the classic graph modularity $Q_{\text{hard}}$ \citep{newman2006modularity} is defined as
\AL{Q_{\text{hard}}(h):= \frac{1}{|\CAL E|} \sum_{i,j\in\CAL I} \Big(A_{i,j} - \gamma\frac{k_i k_j}{|\CAL E|}\Big) 1_{[h_i=h_j]},}
where $\frac{k_i k_j}{|\CAL E|}$ is the expected number of edges between $i,j$ in a random graph under the Newman--Girvan null model theory \citep{newman2004finding}, and $\gamma>0$ denotes the clustering resolution. Here, we adopt the duplicate representation of the undirected graph $\CAL G$, i.e., $(i,j)$ and $(j,i)$ are both in $\CAL E$.

However, we cannot directly use $Q_{\text{hard}}$ to optimize a soft membership $P$ because $1_{[h_i=h_j]}$ is not a differentiable operation. To address the non-differentiability, we propose to maximize the expected modularity under distribution $P$ as a differentiable objective, which we show admits a simple closed form.
\begin{PRP}[closed form]\label{PRP:Qsoft} The expected modularity is
\AL{
Q_{\textnormal{soft}}(P):={}&\ExpOp_{h\sim P}[Q_{\textnormal{hard}}(h)]=\bigg(\frac{1}{|\CAL E|} \sum_{(i,j)\in\CAL E}p_i^\top p_j\bigg)-\gamma \frac{\|P^\top k\|_2^2}{|\CAL E|^2}.\label{eq:q-soft-final}
}
\end{PRP}

We call this objective $Q_{\text{soft}}$ the \emph{soft modularity}. It can be computed efficiently on GPUs when $\CAL E$ and $P$ are sparse, as shown in \PRPref{complexity-qsoft}. 

\begin{PRP}[nearly linear complexity]\label{PRP:complexity-qsoft}
If each row of $P$ has $\le\rho$ nonzero entries, then $Q_\textnormal{soft}(P)$ 
can be computed in time
\AL{\OP O(\rho|\CAL E|)=\OP O(\rho M\log M).}
\end{PRP}


The soft modularity objective \( Q_{\text{soft}} \) can be interpreted as follows. The inner product \( p_i^\top p_j = \sum_{a=1}^C p_{i,a} p_{j,a} \) represents the probability that items \( i \) and \( j \) belong to the same interest prototype under the soft membership distribution ${P}$, serving as a continuous analog to the term $1_{[h_i=h_j]}$ in hard modularity. The intuition behind \( Q_{\text{soft}} \) is to encourage high co-membership probabilities \( p_i^\top p_j \) for item pairs \( (i, j) \) that are heavily co-engaged, while penalizing over-membership to item pairs that are not expected to be connected according to the Newman--Girvan null model theory \citep{newman2004finding}. Furthermore, unlike the original modularity $Q_{\text{hard}}$, our proposed $Q_{\text{soft}}$ is fully differentiable w.r.t.\ $P$, making it suitable as an objective function for optimizing the item interest profiles $P$. 

\subsection{Interest Profile Tokenization for Generative Sequential Recommendation}\label{ssec:method-rec}
In this subsection, we describe how to incorporate item interest profiles into generative sequential recommendation.

\vspace{0.5em}\noindent\textbf{Interest profile tokenization.} As soft item interest profiles are continuous, we use continuous tokens to represent interest profiles. Before introducing our interest profile tokenization, we first introduce the tokenization of item interest prototypes. For each interest prototype $a\in\{1,\dots,C\}$, we define the interest prototype embedding $v_a\in\BB R^d$ of $a$ as the weighted average of the embeddings of the items of interest prototype $a$:
\AL{v_a:=\frac{\sum_{i\in\CAL I}p_{i,a}x_i}{\sum_{i\in\CAL I}p_{i,a}},\qquad a\in\{1,\dots,C\}.\label{eq:prot-emb}}
Let $V:=[v_1,\dots,v_C]^\top\in\BB R^{C\times d}$ denote the interest prototype embedding matrix. Then, Equation~\eqref{eq:prot-emb} can be rewritten into the matrix form as $V=\frac{P^\top X}{P^\top1}$ and can be computed efficiently as $P$ is sparse.

Next, we define the \textit{interest profile token} $y_i$ of each item $i\in\CAL I$ as the weighted average of embeddings $v_a$ of prototypes $a$ of item $i$:
\AL{y_i:=\sum_{a=1}^Cp_{i,a}v_a,\qquad i\in\CAL I.\label{eq:prof-emb}}
Let $Y:=[y_1,\dots,y_{|\CAL I|}]^\top\in\BB R^{|\CAL I|\times d}$ denote the interest profile token matrix. Then, Equation~\eqref{eq:prof-emb} can be rewritten into the matrix form as $Y=PV$, which can be computed efficiently as $P$ is sparse.

\vspace{0.5em}\noindent\textbf{Sequence formulation.} Our key idea is to design a new sequence format that takes item interest profiles into account. Specifically, given the interaction history sequence $\CAL I_u=[i_1,\dots,i_N]$ of a user $u\in\CAL U$, we define a \textit{user interest transition sequence} as follows:
\AL{\CAL R_u:=[\langle\texttt{BOS}\rangle,x_{i_1},y_{i_1},\dots,x_{i_N},y_{i_N}],}
where $\langle\texttt{BOS}\rangle$ denotes a beginning sequence. Here, $y_{i_1},\dots,y_{i_N}$ can be interpreted as the graph-based semantic tokenization of the items $i_1,\dots,i_N$, respectively.

\vspace{0.5em}\noindent\textbf{Training objectives.} Let $F$ denote the generative sequential recommender that autoregressively predicts the next-token distribution. Our goal here is to predict the next item $i_{t+1}\in\CAL I$ that user $u$ will interact with next, given the user interaction history sequence $[i_1,\dots,i_t]$ ($t=1,2,\ldots$). Hence, for each $t=1,\dots,N-1$, we use the cross entropy loss w.r.t.\ to train the generative recommender $F$:
\AL{\CAL L_\text{item}^t:={}&{-\log F(i_{t+1}\mid \langle\texttt{BOS}\rangle,x_{i_1},y_{i_1},\dots,x_{i_t},y_{i_t})}\\
={}&{-\log F(i_{t+1}\mid(\CAL R_u)_{\le 2t+1})}
.}
Furthermore, to help the recommender understand user interest transition behaviors, we also train the recommender to predict the interest of each interaction. Since the ground-truth user interests are not available, we instead propose using the item interest profile as the soft label in the cross entropy loss for each $t=1,\dots,N$:
\AL{\CAL L_\text{profile}^t:={}&{-\sum_{a=1}^C p_{i_t,a}\log F(a\mid \langle\texttt{BOS}\rangle,x_{i_1},y_{i_1},\dots,x_{i_t})}\\
={}&{-\sum_{a=1}^C p_{i_t,a}\log F(a\mid (\CAL R_u)_{\le 2t})}
.}
Together, we train the recommender $F$ using a weighted combination of the two losses simultaneously:
\AL{\CAL L^t:=\CAL L_\text{item}^t+\lambda\CAL L_\text{profile}^t,}
where $\lambda>0$ is a hyperparameter to control the weight of $\CAL L_\text{profile}^t$.



\nocite{wei2024robust,bao2025latte,chen2024wapiti,liu2025breaking,liu2024logic,liu2024class,liu2024aim,liu2023topological,zeng2025pave,zeng2024graph,lin2025moralise,lin2024backtime,qiu2025saffron,qiu2025ask,qiu2025efficient,qiu2024tucket,qiu2023reconstructing,qiu2022dimes,xu2024discrete,li2025model,zou2025transformer,qiu2024gradient,yoo2025embracing,yoo2025generalizable,yoo2024ensuring,chan2024group,wu2024fair,he2024sensitivity,wang2023networked}

\begin{table}[t]
\centering
\caption{Statistics of public datasets.}
\label{tab:datasets}
\begin{tabular}{l|cccc}
\toprule
\textbf{Dataset} & \#Users & \#Items & \#Interacts & Sparsity \\
\midrule
Beauty & 22,363 & 12,101 & 198,502 & 99.93\% \\
Sports & 25,598 & 18,357 & 296,337 & 99.95\% \\
Toys & 19,412 & 11,924 & 167,597 & 99.93\% \\
Yelp & 30,431 & 20,033 & 316,354 & 99.95\% \\
\bottomrule
\end{tabular}
\end{table}

\begin{table*}[t]
\caption{Comparison with baseline methods on public datasets (best marked in bold). Our proposed \Ours{} consistently outperforms all baseline methods on all four datasets.}
\label{tab:main}
\vspace{-0.5em}
\centering
\resizebox{1\linewidth}!{
\begin{tabular}{cl|cc|ccccc|ccc}
\toprule
\multicolumn{2}{c|}{\textbf{Method Type}\,$\rightarrow$} & \multicolumn{2}{c|}{Classic} & \multicolumn{5}{c|}{Sequential / Tokenization-Based} & \multicolumn{3}{c}{Graph-Based} \\
\textbf{Dataset}\,$\downarrow$ & \textbf{Metric}\,$\downarrow$ &{POP}&{MF}&{GRU4Rec}&{SASRec}&{BERT4Rec}&{Caser}&{EAGER}&{LightGCN}&{HeLLM}&\Ours{} (Ours)\\
\midrule
\multirow{6}*{Beauty}
&Recall@1 &{0.0678}&{0.0405}&{0.1870}&{0.1531}&{0.1337}&{0.1337}&{0.1213}&{0.1435}&{0.1690}&\textbf{0.2067}\\
&Recall@5 &{0.2105}&{0.1461}&{0.3741}&{0.3640}&{0.3032}&{0.3125}&{0.2678}&{0.3081}&{0.2802}&\textbf{0.3917}\\
&Recall@10&{0.3386}&{0.2311}&{0.4696}&{0.4739}&{0.3942}&{0.4106}&{0.3663}&{0.4042}&{0.3413}&\textbf{0.4892}\\
&NDCG@5   &{0.1391}&{0.0934}&{0.2848}&{0.2622}&{0.2219}&{0.2268}&{0.1962}&{0.2286}&{0.2278}&\textbf{0.3035}\\
&NDCG@10  &{0.1803}&{0.1207}&{0.3156}&{0.2975}&{0.2512}&{0.2584}&{0.2278}&{0.2596}&{0.2474}&\textbf{0.3334}\\
&MRR      &{0.1558}&{0.1096}&{0.2852}&{0.2614}&{0.2263}&{0.2308}&{0.2060}&{0.2340}&{0.2359}&\textbf{0.3034}\\
\midrule
\multirow{6}*{Sports}
&Recall@1 &{0.0763}&{0.0489}&{0.1455}&{0.1255}&{0.1135}&{0.1160}&{0.0417}&{0.1226}&{0.1021}&\textbf{0.1750}\\
&Recall@5 &{0.2293}&{0.1603}&{0.3466}&{0.3375}&{0.2866}&{0.3055}&{0.1398}&{0.2841}&{0.2214}&\textbf{0.3903}\\
&Recall@10&{0.3423}&{0.2491}&{0.4622}&{0.4722}&{0.4014}&{0.4299}&{0.2232}&{0.3855}&{0.2946}&\textbf{0.5093}\\
&NDCG@5   &{0.1538}&{0.1048}&{0.2497}&{0.2341}&{0.2020}&{0.2126}&{0.0906}&{0.2056}&{0.1641}&\textbf{0.2869}\\
&NDCG@10  &{0.1902}&{0.1334}&{0.2869}&{0.2775}&{0.2390}&{0.2527}&{0.1174}&{0.2383}&{0.1875}&\textbf{0.3254}\\
&MRR      &{0.1660}&{0.1202}&{0.2520}&{0.2378}&{0.2100}&{0.2191}&{0.1078}&{0.2133}&{0.1750}&\textbf{0.2867}\\
\midrule
\multirow{6}*{Toys}
&Recall@1 &{0.0585}&{0.0257}&{0.1878}&{0.1262}&{0.1114}&{0.0997}&{0.1201}&{0.1310}&{0.1507}&\textbf{0.2006}\\
&Recall@5 &{0.1977}&{0.0978}&{0.3682}&{0.3344}&{0.2614}&{0.2795}&{0.2717}&{0.2799}&{0.2717}&\textbf{0.3779}\\
&Recall@10&{0.3008}&{0.1715}&{0.4663}&{0.4493}&{0.3540}&{0.3896}&{0.3777}&{0.3721}&{0.3375}&\textbf{0.4691}\\
&NDCG@5   &{0.1286}&{0.0614}&{0.2820}&{0.2327}&{0.1885}&{0.1919}&{0.1977}&{0.2078}&{0.2147}&\textbf{0.2931}\\
&NDCG@10  &{0.1618}&{0.0850}&{0.3136}&{0.2698}&{0.2183}&{0.2274}&{0.2319}&{0.2374}&{0.2359}&\textbf{0.3225}\\
&MRR      &{0.1430}&{0.0819}&{0.2842}&{0.2338}&{0.1967}&{0.1973}&{0.2075}&{0.2154}&{0.2228}&\textbf{0.2942}\\
\midrule
\multirow{6}*{Yelp}
&Recall@1 &{0.0801}&{0.0624}&{0.2375}&{0.2405}&{0.2188}&{0.2053}&{0.0976}&\textbf{0.2696}&{0.2148}&{0.2558}\\
&Recall@5 &{0.2415}&{0.2036}&{0.5745}&{0.5976}&{0.5111}&{0.5437}&{0.2903}&{0.5517}&{0.4053}&\textbf{0.6105}\\
&Recall@10&{0.3609}&{0.3153}&{0.7373}&{0.7597}&{0.6661}&{0.7265}&{0.4088}&{0.6756}&{0.4909}&\textbf{0.7736}\\
&NDCG@5   &{0.1622}&{0.1333}&{0.4113}&{0.4252}&{0.3696}&{0.3784}&{0.1960}&{0.4165}&{0.3160}&\textbf{0.4398}\\
&NDCG@10  &{0.2007}&{0.1692}&{0.4642}&{0.4778}&{0.4198}&{0.4375}&{0.2343}&{0.4566}&{0.3436}&\textbf{0.4927}\\
&MRR      &{0.1740}&{0.1470}&{0.3927}&{0.4026}&{0.3595}&{0.3630}&{0.2009}&{0.4025}&{0.3142}&\textbf{0.4180}\\
\midrule
\multicolumn{2}{c|}{\textbf{Average Rank}\,$\rightarrow$} & {8.62} & {9.75} & {2.46} & {2.83} & {6.35} & {5.23} & {7.94} & {4.50} & {6.27} & \textbf{1.04} \\ 
\bottomrule
\end{tabular}
}
\end{table*}

\section{Industrial Deployment}\label{ssec:method-deploy}
We have successfully productionized our \Ours{} on multiple product surfaces\AtMeta{}. Considering the huge user base (billions of monthly active users~\citep{nyt}) of Meta products (e.g., Instagram Reels), a main challenge in deployment lies in scalability and real-time efficiency.
To ensure the timeliness of item interest profiles, we have implemented an in-house high performance graph processing engine that can efficiently execute the graph clustering algorithm~\citep{louvain} in a distributed environment.
We run the graph clustering algorithm periodically offline, so the execution time of the graph clustering algorithm has no impact on the response time to users.


\section{Experiments}
We conduct extensive experiments on both public datasets and our product surfaces to answer the following research questions:
\begin{enumerate}
\renewcommand\labelenumi{\textbf{RQ\theenumi:}}
\item\label{RQ:baseline}How does the proposed \Ours{} compare with state-of-the-art recommendation methods offline?
\item\label{RQ:product}How does the proposed \Ours{} perform on \Metas{} product surfaces online?
\item\label{RQ:eff}How does the proposed \Ours{} influence the efficiency in training and inference, respectively?
\item\label{RQ:abla}How do the proposed soft graph clustering and item interest profiling influence recommendation quality?
\end{enumerate}

\subsection{Experimental Settings}
In this subsection, we describe the experimental settings in our offline empirical evaluation on public datasets. 

\vspace{0.5em}\noindent\textbf{Public datasets.} To evaluate our proposed \Ours{}, we conduct experiments on four widely-used public real-world datasets: Beauty, Sports, Toys, and Yelp. Beauty, Sports, and Toys are sub-categories from the Amazon review data collection \citep{mcauley2015image}, which is collected from Amazon, a popular online shopping platform. Yelp is another popular online platform that provides crowd-sourced reviews about businesses, and the Yelp dataset \citep{asghar2016yelp} is collected from Yelp. In our experiments, we use Yelp data from January 1st, 2019 onwards. The dataset statistics are summarized in Table~\ref{tab:datasets}. We can see that all of these datasets are highly sparse, which ensures that these datasets are sufficiently challenging for evaluating the performance under the sequential recommendation setting.

\vspace{0.5em}\noindent\textbf{Baseline methods.} To comprehensively benchmark the performance of our proposed \Ours{}, we compare our \Ours{} with six strong baseline methods, including classic methods, sequential / tokenization-based methods, and graph-based methods. We briefly describe the baseline methods as follows. (i) Classic methods: \textbf{POP} is a heuristic method that ranks items based on their popularity, which is known to be a strong baseline in recommendation (e.g., \citet{yoo2025generalizable}). Matrix factorization \textbf{MF} \citep{koren2009matrix} is the most classic recommendation method; we use the Bayesian Personalized Ranking (BPR) loss \citep{rendle2012bpr} here. (ii) Sequential methods: \textbf{GRU4Rec} \citep{hidasi2015session} utilizes the GRU recurrent neural network \citep{cho2014properties} for modeling sequences and subsequently makes recommendation predictions. \textbf{SASRec} \citep{kang2018self} uses multi-head self-attention \citep{vaswani2017attention} to handle intricate sequential data. \textbf{BERT4Rec} \citep{sun2019bert4rec} uses the cloze objective function from BERT \citep{devlin2019bert} to enable self-supervised learning instead of autoregressively predicting only the next item. \textbf{Caser} \citep{tang2018personalized} combines both horizontal and vertical convolution operations to more effectively capture complex interactions within user interaction history sequences. \textbf{EAGER} \citep{wang2024eager} is a two-stream generative recommender with behavior-semantic collaboration. (iii) Graph-based methods. \textbf{LightGCN} \citep{he2020lightgcn} is a strong graph convolutional network for collaborative filtering. \textbf{HeLLM} \citep{guo2025multi} uses hypergraphs to enhance LLM-based recommenders.

\vspace{0.5em}\noindent\textbf{Evaluation protocol \& metrics.} For sequential recommendation evaluation, we sort user interactions by their timestamps in ascending order to form user interaction history sequences. To ensure a reliable evaluation, we employ a 5-core setting to exclude items with fewer than 5 interactions, following previous works (e.g., \citet{rendle2010factorizing,sun2019bert4rec}). Following previous works (e.g., \citet{kang2018self,ren2020sequential}), we adopt the leave-one-out approach to split the dataset. Specifically, for the interaction history sequence of each user, we used the last item as the test data, the second-to-last item as the validation data, and the rest of the sequence as the training data. To evaluate the ranking capability of our proposed \Ours{}, we form the test set as follows: for each user, we use the last item they interacted with as the positive item and randomly sample 99 items from the remaining items as negative items. We use Recall@1, Recall@5, Recall@10, NDCG@5, NDCG@10, and the mean reciprocal rank (MRR) as metrics.

\vspace{0.5em}\noindent\textbf{Implementation details.}
We use Llama 2 13B as the recommender model backbone and finetune it for 3 epochs using low-rank adapters (LoRA) with 16 ranks and rank dropout rate 0.05. We initialize the model parameters using the Xavier initialization and optimize the model parameters using the Adam optimizer with initial learning rate 0.0003 and the cosine learning rate schedule with 100 warmup steps. We use SASRec embeddings with $d=64$ as item embeddings and limit the maximum sequence length to 50 for all methods on all datasets to avoid out-of-memory errors and to ensure a fair comparison. To control the number of interest prototypes, we use $\gamma=0.8$ for Beauty and Toys and $\gamma=1$ for Sports and Yelp.
For baseline methods, we use the open-source code released by their authors and adapt the code to our experimental settings.

\subsection{Offline Testing}
To answer RQ\ref{RQ:baseline}, we compare our proposed \Ours{} with baseline methods to evaluate the recommendation performance. The results presented in Table~\ref{tab:main} and Figure~\ref{fig:recall-k} demonstrate the effectiveness of our proposed method \Ours{} in comparison to all six baseline methods across all four public datasets. Our \Ours{} consistently outperforms all baseline methods across all six metrics on all four datasets, achieving the highest recall, NDCG, and MRR. We observe that the performance gap between our \Ours{} and the best-performing baseline method is significant. For instance, our proposed \Ours{} achieves 14.9\% higher NDCG@5 than the best-performing baseline on the Sports dataset.
As these datasets are large-scale datasets with diverse user behavior patterns, the results suggest that our method is capable of handling complex user behavior data and can effectively capture the underlying patterns. In terms of average rank, our \Ours{} achieves the top rank across all datasets, indicating its overall superiority over the baseline methods; in stark contrast, the average rank of the baseline methods varies across datasets.
Overall, the results demonstrate the effectiveness of our proposed \Ours{} in modeling user behavior patterns and item relationships using the co-engagement graph. This suggests that our \Ours{} can model item relationships and user behavior patterns using the co-engagement graph and provides a substantial advantage over existing methods.

\begin{figure*}[t]
\centering

\begin{subfigure}[b]{0.24\linewidth}
\centering
\includegraphics[width=\linewidth]{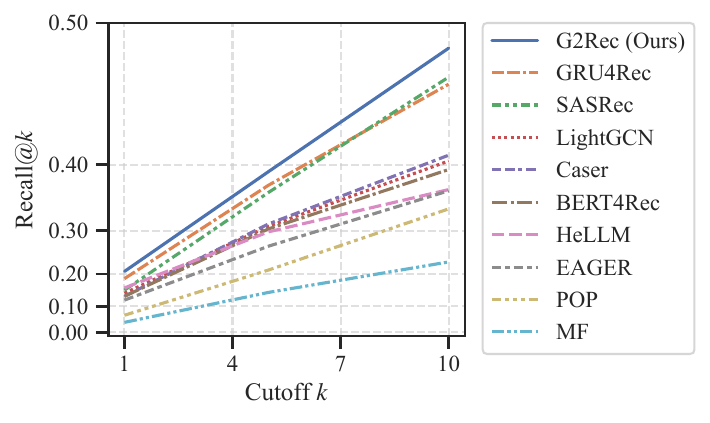}
\caption{Beauty.}
\label{fig:recall-B}
\end{subfigure}
\hfill
\begin{subfigure}[b]{0.24\linewidth}
\centering
\includegraphics[width=\linewidth]{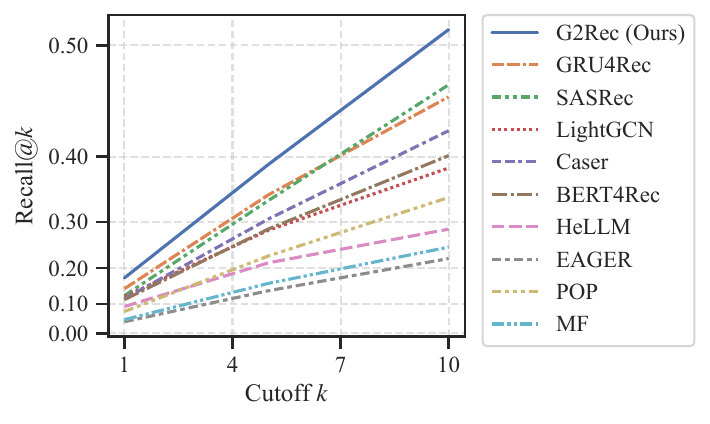}
\caption{Sports.}
\label{fig:recall-S}
\end{subfigure}
\hfill
\begin{subfigure}[b]{0.24\linewidth}
\centering
\includegraphics[width=\linewidth]{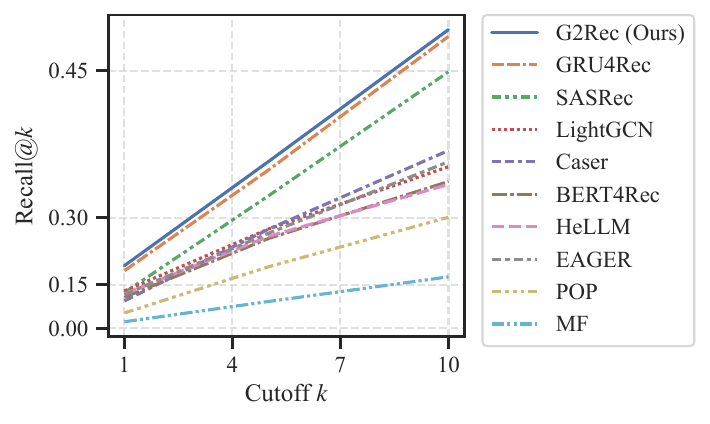}
\caption{Toys.}
\label{fig:recall-T}
\end{subfigure}
\hfill
\begin{subfigure}[b]{0.24\linewidth}
\centering
\includegraphics[width=\linewidth]{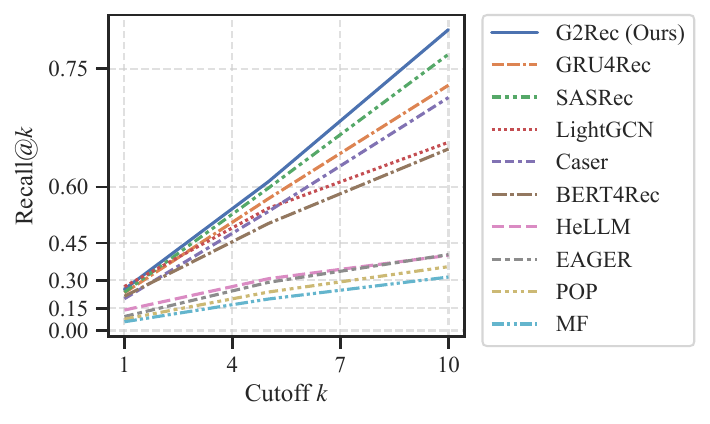}
\caption{Yelp.}
\label{fig:recall-Y}
\end{subfigure}
\vspace{-0.5em}
\caption{Our \Ours{} consistently outperforms all baselines on all four datasets under various cutoffs.}
\label{fig:recall-k}
\end{figure*}

\begin{table}[t]
\centering
\caption{Training and inference time per batch. Our interest profile tokens have only negligible impact on training and inference time.}
\label{tab:eff}
\vspace{-0.9em}
\begin{tabular}{l|cc}
\toprule
\textbf{Method}&Training&Inference\\
\midrule
Item-Only & 0.763 & 0.0534 \\
\Ours{} (Ours)   & 0.806 & 0.0561 \\
\midrule
Overhead & +0.043 & +0.0027 \\
\bottomrule
\end{tabular}
\end{table}

\subsection{Large-Scale Online Testing}

We have productized our proposed \Ours{} on multiple product surfaces\AtMeta{}. While some internal modifications were made to adapt it for deployment on each product surface, the core methodology remains consistent with the proposed framework. 

To answer RQ\ref{RQ:product}, we evaluate our proposed \Ours{} via both short-term (7-day) and long-term online A/B testing. (Due to the corporate policy, we are unable to disclose the detailed settings.) We have observed \textbf{consistent improvement} in terms of a wide range of top-line metrics on \textbf{user engagement}, \textbf{content diversity}, and \textbf{serving efficiency}. In particular, 
product launches with our holistic interest modeling lead to $>0.03\%$ in-session improvement and exhibit solid wins ($0.06\%$ to $0.19\%$) in terms of multiple user engagement metrics such as total time-spent, likes, shares, etc. 
These results provide strong evidence that our proposed \Ours{} considerably enhances user experience by accurately capturing user behavioral patterns and interest prototypes.

\subsection{Efficiency Study}
In addition to evaluating the effectiveness of our proposed \Ours{}, we also investigate its efficiency as efficiency is important in industrial scenarios. Table~\ref{tab:eff} reports the training and inference time per batch for our method and the item-only baseline on the Beauty dataset.
The results show that incorporating interest profile tokens into the model has only a negligible impact on both training and inference time. Specifically, our method requires only 0.043 seconds more per batch during training and only 0.0027 seconds more per batch during inference compared to the item-only baseline. This suggests that our approach can be efficiently integrated into industrial recommendation systems without incurring significant computational overhead.
This efficiency is due to the fact that our interest profile tokens are computed offline using the co-engagement graph, and then simply concatenated with the item embeddings during training and inference. This design allows us to leverage the structural information captured by the co-engagement graph without requiring expensive online computations.
Overall, these results demonstrate that our approach offers a good balance between effectiveness and efficiency, making it a practical solution for industrial recommender systems.

\begin{table}[t]
\centering
\caption{Ablation study on soft graph clustering. Our soft graph clustering algorithm consistently outperforms hard graph clustering on all datasets in terms of modularity.}
\label{tab:abla-clus}
\vspace{-0.5em}
\begin{tabular}{l|cccc}
\toprule
\textbf{Method}&Beauty&Sports&Toys&Yelp\\
\midrule
Hard (Leiden) & 0.419 & 0.365 & 0.437 & 0.691 \\
\textbf{Soft} (Ours) & \textbf{0.499} & \textbf{0.452} & \textbf{0.550} & \textbf{0.757} \\
\bottomrule
\end{tabular}
\end{table}

\begin{figure}
\centering
\includegraphics[width=0.5\linewidth]{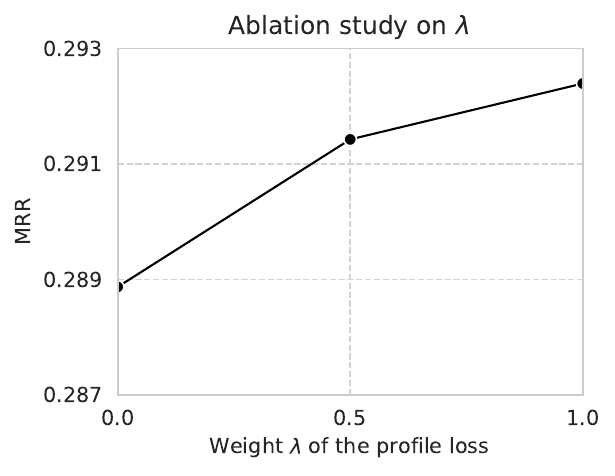}
\caption{Ablation study on the weight $\lambda$ of $\CAL L^t_\text{profile}$. This confirms the efficacy of our proposed interest profiling.}
\label{fig:coef}
\end{figure}

\subsection{Ablation Studies}
To answer RQ\ref{RQ:abla}, we conduct ablation studies on our proposed soft graph clustering and item interest profiling.

\vspace{0.5em}\noindent\textbf{Efficacy of soft graph clustering.} We conduct an ablation study to evaluate the efficacy of our soft graph clustering algorithm. We compare our method with hard graph clustering produced by the Leiden algorithm \citep{leiden}. As shown in Table \ref{tab:abla-clus}, the results demonstrate that our soft graph clustering algorithm consistently outperforms hard graph clustering in terms of modularity, a widely-used metric for evaluating the quality of graph clustering. The improvements are significant across all datasets (for example, 23.8\% improvement in modularity score on the Sports dataset). This suggests that our soft graph clustering algorithm is able to capture more nuanced and complex item interest profiles in the co-engagement graph, which is essential for accurately modeling user interests. These results validate the design choice of using soft graph clustering in our proposed \Ours{} and highlight its importance in achieving high-quality item interest profiles.

\vspace{0.5em}\noindent\textbf{Efficacy of item interest profiling.} To verify the efficacy of our interest profile loss $\CAL L^t_\text{profile}$, we conduct an ablation study to investigate the impact of $\CAL L^t_\text{profile}$ on the overall recommendation performance of our propsed \Ours{}. Specifically, we compare the MRRs under various weights $\lambda$ on the Beauty dataset. The results are shown in Figure~\ref{fig:coef}. As can be seen from the figure, as the weight $\lambda$ increases, the recommendation metric MRR also improves. This demonstrates the efficacy of our proposed interest profile loss $\CAL L^t_\text{profile}$ in capturing the complex relationships between items and enhancing the accuracy of recommendations. Notably, when $\lambda=0$ (which means that the interest profile loss is not used), the MRR is significantly lower than when $\lambda>0$. This suggests that the interest profile loss plays a crucial role in improving the recommendation performance of our proposed \Ours{}. Overall, our ablation study demonstrates the importance of the interest profile loss in our proposed \Ours{} and highlights its efficacy in improving the quality of recommendations in real-world applications.


\section{Related Work}


\subsection{Generative Recommendation}\label{ssec:rel-genrec}

Generative recommendation formulates recommendation tasks as sequential generation problems, aiming to predict user preferences by generating sequences of recommended items rather than simply ranking predefined candidates~\citep{deng2025onerec}. Recent studies have employed deep generative models, particularly autoregressive language models, to directly rank candidate items~\citep{hou2024large}, predict user ratings~\citep{gao2023chat}, and retrieve relevant information~\citep{zhu2023large}.
P5~\citep{geng2022recommendation} conceptualizes recommendation as a language processing task, integrating various recommendation scenarios within a unified sequence-to-sequence framework. This approach effectively leverages multiple sources of information, facilitating richer and more nuanced user modeling. However, bridging the gap between natural language processing and discrete user and item indexing remains challenging. \citet{hua2023index} explored several indexing methods, including sequential indexing, collaborative indexing, semantic indexing, and hybrid indexing, to incorporate indexing mechanisms into generative models. More recent studies~\citep{xiao2025progressive,wang2024eager,xu2024slmrec,li2023e4srec} incorporated both semantic and collaborative information into index tokenization. Despite demonstrating promising outcomes, these approaches continue to face difficulties in fully capturing complex user behavioral patterns and detailed item semantic information.

\subsection{Graphs in Recommendation}\label{ssec:rel-graph}
Graph-based recommendation approaches leverage the inherent relational structure among users and items to capture complex connectivity patterns that traditional methods often overlook. 
Recent advances in GNNs have propelled significant progress in this area, allowing models to integrate both structural and feature-based information into embedding learning. 
Notably, Graph Convolutional Matrix Completion (GCMC)~\citep{van2017graph} introduces a graph auto-encoder framework that facilitates differentiable message passing over bipartite user-item graphs.
Neural Graph Collaborative Filtering (NGCF)~\citep{wang2019neural} explored learning user and item embeddings by propagating neighbor information in user-item graphs.  LightGCN~\citep{he2020lightgcn} then enhanced this idea by adopting linear neighbor information aggregation. 
Despite their effectiveness, GNN-based methods typically perceive only a small subgraph centered around individual users or items, limiting their capacity to fully exploit holistic graph information and potentially overlooking valuable global structural insights. 
More recently, researchers have begun integrating LLMs with GNNs to enhance recommendation performance by leveraging the relational modeling capabilities of GNNs and the natural language understanding strengths of LLMs. Existing approaches generally fall into three categories: graph-augmented LLM~\citep{ma2024triple,wang2024enhancing,ma2024xrec}, LLM-augmented graph~\citep{wang2023enhancing,liu2025understanding,jeon2025topic}, and LLM-graph collaboration~\citep{runfeng2023lkpnr,zhao2024dynllm}.
\section{Conclusion \& Future Work}
In this paper, we introduced \Ours{}, a scalable item schema designed to enhance generative sequential recommendation systems by effectively integrating holistic graph-based co-engagement modeling for semantic tokenization. Our approach addresses the limitations of existing methods, which often struggle with scalability and the accurate incorporation of complex user-behavioral and item-semantic information. By constructing an item-item co-engagement graph and employing a soft clustering algorithm, \Ours{} derives interest prototypes that transform user interaction history sequences into the more informative item interest profiles. This transformation allows the recommendation model to capture comprehensive and semantic user interest profiles without relying on predefined ground-truth interests.
Our extensive experiments on public datasets, along with successful online deployment on product surfaces, demonstrate the superiority of \Ours{} over current state-of-the-art methods. The results highlight that \Ours{} have the ability to provide more accurate and holistic modeling of user behavior, ultimately leading to improved recommendation performance in industrial applications. \Ours{} represents a significant step forward in the field of generative sequential recommendation, offering a scalable and effective solution for industry-scale applications. Future work may explore further enhancements to the framework, such as incorporating additional data sources or refining the clustering algorithm to capture even more nuanced interests.


\bibliographystyle{assets/plainnat}
\bibliography{paper}

\beginappendix


\section{Proof of \THMref{preserv-lap}}\label{app:proof-preserv-lap}


\paragraph{Setup and conventions.}
We treat the co-engagement graph as \emph{undirected}. Since every user has $N$ interactions, the co-engagements of each user $u\in\CAL U$ form a complete graph (clique) $K_{\CAL I_u}$ on the $N$ items in $\CAL I_u$, with $\binom N2=\frac{N(N-1)}2$ undirected edges. For an undirected edge $\{i,j\}$, write $b_{ij}:=e_i-e_j\in\BB R^{|\CAL I|}$, where $e_i$ denotes the $i$-th standard basis vector; this edge contributes the rank-one term $b_{ij}b_{ij}\Tp$ to the graph Laplacian. Self-loops $(i,i)$ satisfy $b_{ii}=0$ and thus contribute nothing, so we discard them and identify the sampling domain $\CAL I_u\times\CAL I_u$ with the $\binom N2$ undirected edges of $K_{\CAL I_u}$. Accordingly,
\AM{
L_{\CAL E^*}=\sum_{u\in\CAL U}\sum_{\{i,j\}\subseteq\CAL I_u}b_{ij}b_{ij}\Tp,
\qquad
L_{\CAL E}=\frac{N(N-1)}{2m}\sum_{u\in\CAL U}\sum_{\ell=1}^{m}b_{i^u_\ell j^u_\ell}b_{i^u_\ell j^u_\ell}\Tp,
}
where, for each user $u$, the edges $\{i^u_\ell,j^u_\ell\}_{\ell=1}^m$ are drawn independently and uniformly with replacement from the $\binom N2$ edges of $K_{\CAL I_u}$. The reweighting factor $\frac{N(N-1)}{2m}=\binom N2/m$ rescales the $m$ samples back to the total edge weight $\binom N2$ of each user's clique, which (as we verify below) makes $L_{\CAL E}$ an unbiased estimator of $L_{\CAL E^*}$. Without loss of generality, we assume that $\CAL G^*$ is connected; otherwise, the argument below applies verbatim to each connected component.




\paragraph{Auxiliary lemma.} We will use the following Lemma~\ref{LEM:eff-resist}.

\begin{LEM}[effective resistance]\label{LEM:eff-resist}
Let $L^\dagger_{\CAL E^*}$ denote the Moore--Penrose pseudoinverse of $L_{\CAL E^*}$. For every co-engaged pair $\{i,j\}$ (i.e., $\{i,j\}\subseteq\CAL I_u$ for some $u\in\CAL U$),
\AM{
b_{ij}\Tp L^\dagger_{\CAL E^*}b_{ij}\le\frac2N.
}
\end{LEM}
\begin{proof}[Proof of \LEMref{eff-resist}]
The quantity $b_{ij}\Tp L^\dagger_{\CAL E^*}b_{ij}$ is the effective resistance between items $i$ and $j$ in $\CAL G^*$. The clique $K_{\CAL I_u}$ is a subgraph of $\CAL G^*$, and adding edges (here, the remaining co-engagements) can only decrease effective resistance by Rayleigh's monotonicity law. Hence, it suffices to bound the effective resistance within $K_N=K_{\CAL I_u}$ alone. The Laplacian of $K_N$ is $L_{K_N}=NI-\mathbf1\mathbf1\Tp$ on the $N$ vertices of the clique, whose pseudoinverse acts as $\frac1N$ on the subspace orthogonal to $\mathbf1$. Since $b_{ij}\perp\mathbf1$, we have $b_{ij}\Tp L_{K_N}^\dagger b_{ij}=\frac1N\|b_{ij}\|_2^2=\frac2N$. The claim follows.
\end{proof}

\emph{Theorem proof.} We are now ready to prove \THMref{preserv-lap}.

\begin{proof}[Proof of \THMref{preserv-lap}]
First, fix a user $u$ and a single draw $\{i,j\}$ uniform over the $\binom N2$ edges of $K_{\CAL I_u}$. Then
\AM{
\ExpOp\Big[\tfrac{N(N-1)}{2m}\,b_{ij}b_{ij}\Tp\Big]
=\frac{N(N-1)}{2m}\cdot\frac1{\binom N2}\sum_{\{i,j\}\subseteq\CAL I_u}b_{ij}b_{ij}\Tp
=\frac1m\sum_{\{i,j\}\subseteq\CAL I_u}b_{ij}b_{ij}\Tp.
}
Summing over the $m$ independent draws of user $u$ and then over all users gives $\ExpOp[L_{\CAL E}]=L_{\CAL E^*}$.

Enumerate all $n:=|\CAL U|m$ draws as $\ell=1,\dots,n$, and let $X_\ell:=\frac{N(N-1)}{2m}b_{i_\ell j_\ell}b_{i_\ell j_\ell}\Tp$ be the reweighted contribution of the $\ell$-th draw, so that $L_{\CAL E}=\sum_{\ell=1}^n X_\ell$ and the $X_\ell$ are independent. Define
\AM{
Y_\ell:=(L^\dagger_{\CAL E^*})^{1/2}X_\ell(L^\dagger_{\CAL E^*})^{1/2},
\qquad
\bar Y:=\sum_{\ell=1}^n Y_\ell=(L^\dagger_{\CAL E^*})^{1/2}L_{\CAL E}(L^\dagger_{\CAL E^*})^{1/2}.
}
Let $\Pi:=(L^\dagger_{\CAL E^*})^{1/2}L_{\CAL E^*}(L^\dagger_{\CAL E^*})^{1/2}$ be the orthogonal projection onto $\OP{range}(L_{\CAL E^*})$. By unbiasedness, $\ExpOp[\bar Y]=\Pi$, and $\Pi=I$ on $\OP{range}(L_{\CAL E^*})$.

Write $w:=\frac{N(N-1)}{2m}$. Each $Y_\ell\succeq0$, and by \LEMref{eff-resist},
\AM{
\|Y_\ell\|_2=w\,b_{ij}\Tp L^\dagger_{\CAL E^*}b_{ij}\le\frac{N(N-1)}{2m}\cdot\frac2N=\frac{N-1}m\le\frac Nm=:R.
}

Since $b_{ij}\Tp L^\dagger_{\CAL E^*}b_{ij}\le\frac2N$, each $Y_\ell$ satisfies $Y_\ell^2=(w\,b_{ij}\Tp L^\dagger_{\CAL E^*}b_{ij})\,Y_\ell\preceq\frac Nm Y_\ell$. Taking expectations and summing, and using $\ExpOp[Z_\ell^2]\preceq\ExpOp[Y_\ell^2]$ for the centered matrices $Z_\ell:=Y_\ell-\ExpOp[Y_\ell]$,
\AM{
v:=\Big\|\sum_{\ell=1}^n\ExpOp[Z_\ell^2]\Big\|_2
\le\Big\|\sum_{\ell=1}^n\ExpOp[Y_\ell^2]\Big\|_2
\le\frac Nm\Big\|\sum_{\ell=1}^n\ExpOp[Y_\ell]\Big\|_2
=\frac Nm\|\Pi\|_2=\frac Nm.
}

The matrices $Z_\ell=Y_\ell-\ExpOp[Y_\ell]$ are independent, symmetric, and centered, of size $d=|\CAL I|$, with $\|Z_\ell\|_2\le R=\frac Nm$ and variance proxy $v\le\frac Nm$. By the matrix Bernstein inequality,
\AM{
\Prb\{\|\bar Y-\Pi\|_2\ge\epsilon\}\le2|\CAL I|\exp\!\Big({-\frac{\epsilon^2/2}{\frac Nm+\frac\epsilon3\frac Nm}}\Big).
}

The right-hand side is at most $\delta$ whenever
\AM{
\frac{\epsilon^2/2}{\frac Nm(1+\frac\epsilon3)}\ge\log\frac{2|\CAL I|}\delta
\iff
m\ge\frac{2N(1+\epsilon/3)}{\epsilon^2}\log\frac{2|\CAL I|}\delta
=2N\Big(\frac1{\epsilon^2}+\frac1{3\epsilon}\Big)\log\frac{2|\CAL I|}\delta.
}
The choice $m=\big\lceil2N(\frac1{3\epsilon}+\frac1{\epsilon^2})\log\frac{2|\CAL I|}\delta\big\rceil$ satisfies this, so $\|\bar Y-\Pi\|_2\le\epsilon$ with probability at least $1-\delta$.

Finally, the inequality $\|\bar Y-\Pi\|_2\le\epsilon$ is equivalent to $(1-\epsilon)\Pi\preceq\bar Y\preceq(1+\epsilon)\Pi$. Conjugating by $L_{\CAL E^*}^{1/2}$ and using $L_{\CAL E^*}^{1/2}\Pi L_{\CAL E^*}^{1/2}=L_{\CAL E^*}$ together with $L_{\CAL E^*}^{1/2}\bar Y L_{\CAL E^*}^{1/2}=L_{\CAL E}$ (valid since $\OP{range}(L_{\CAL E})\subseteq\OP{range}(L_{\CAL E^*})$), we obtain
\AM{
(1-\epsilon)L_{\CAL E^*}\preceq L_{\CAL E}\preceq(1+\epsilon)L_{\CAL E^*}.
}

Since every user contributes $m$ sampled edges and $M=|\CAL U|N$,
\AM{
|\CAL E|=|\CAL U|\,m=\OP O\Big(|\CAL U|\,N\log\tfrac{2|\CAL I|}\delta\Big)=\OP O\Big(M\log\tfrac{2|\CAL I|}\delta\Big)=\OP O(M\log M),
}
where we used $|\CAL I|\le M$ and treated $\epsilon,\delta$ as constants. This completes the proof.
\end{proof}

\section{Proof of Proposition~\ref{PRP:Qsoft}}

\begin{proof}[Proof]
By the linearity of expectation,
\AL{&Q_{\text{soft}}(P)={}\ExpOp_{h\sim P}[Q_{\text{hard}}(h)]
\\={}&\ExpOp_{h\sim P}\bigg[\frac{1}{|\CAL E|} \sum_{i,j\in\CAL I} \Big(A_{i,j} - \gamma\frac{k_i k_j}{|\CAL E|}\Big) 1_{[h_i=h_j]}\bigg]
\\={}&
\frac{1}{|\CAL E|} \sum_{i,j\in\CAL I} \Big(A_{i,j} - \gamma\frac{k_i k_j}{|\CAL E|}\Big)\ExpOp_{h\sim P}[1_{[h_i=h_j]}]
\\={}&\frac{1}{2|\CAL E|} \sum_{i,j\in\CAL I} \Big(A_{i,j} - \gamma\frac{k_i k_j}{2|\CAL E|}\Big)\sum_{a=1}^C p_{i,a}p_{j,a}
\\={}&\frac{1}{|\CAL E|} \sum_{i,j\in\CAL I} \Big(A_{i,j} - \gamma\frac{k_i k_j}{|\CAL E|}\Big)p_i^\top p_j
\\={}&
\bigg(\frac{1}{|\CAL E|} \sum_{(i,j)\in\CAL E}p_i^\top p_j\bigg)-\gamma \frac{\|P^\top k\|_2^2}{|\CAL E|^2}\nonumber
.\qedhere
}
\end{proof}

\section{Proof of Proposition~\ref{PRP:complexity-qsoft}}

\begin{proof}[Proof]
Since each row of $P$ has $\le\rho$ nonzero entries, then $\sum_{(i,j)\in\CAL E}p_i\Tp p_j$ needs time $\sum_{(i,j)\in\CAL E}\OP O(\rho)=\OP{\rho|\CAL E|}$. Since $P$ has $\le\rho|\CAL I|$ nonzero entries, then $\|P\Tp k\|_2^2$ needs time $\OP O(\rho|\CAL I|)\le\OP O(\rho|\CAL E|)$. Hence, the conclusion follows from \THMref{preserv-lap}.
\end{proof}

\end{document}